\documentclass[openacc]{rsproca_new}
\usepackage{graphicx}
\usepackage{epstopdf, epsfig}
\usepackage{amsmath}  
\usepackage{amssymb}
\usepackage{mathtools}
\usepackage{tikz}
\usetikzlibrary{decorations.pathreplacing,angles,quotes}

\newtheorem{theorem}{\bf Theorem}[section]

\newtheorem{corollary}{\bf Corollary}[section]

\newtheorem{remark}{\bf Remark}[section]

\begin{document}
\title{Stability of Two-dimensional Potential Flows Using Bicomplex Numbers}

\author{
V. G. Kleine$^{1,2}$, A. Hanifi$^{1}$ and D. S. Henningson$^{1}$}

\address{$^{1}$KTH Royal Institute of Technology, FLOW, Department of Engineering Mechanics, Stockholm, Sweden\\
$^{2}$Instituto Tecnol\'{o}gico de Aeron\'{a}utica, S\~{a}o Jos\'{e} dos Campos - SP, Brazil}

\subject{fluid mechanics, mechanical engineering, applied mathematics}


\corres{Vitor G. Kleine\\
\email{vitok@mech.kth.se}}

\begin{abstract}
The use of the complex velocity potential and the complex velocity is widely disseminated in the study of two-dimensional incompressible potential flows. The advantages of working with complex analytical functions made this representation of the flow ubiquitous in the field of theoretical aerodynamics. However, this representation is not usually employed in linear stability studies, where the representation of the velocity as real vectors is preferred by most authors, in order to allow the representation of the perturbation as the complex exponential function. Some of the classical attempts to use the complex velocity potential in stability studies suffer from formal errors. In this work, we present a framework that reconciles these two complex representations using bicomplex numbers. This framework is applied to the stability of the von Kármán vortex street and a generalized formula is found. It is shown that the classical results of the  symmetric and staggered von Kármán vortex streets are just particular cases of the generalized dynamical system in  bicomplex formulation.
\end{abstract}




\begin{fmtext}
\section{Introduction} \label{sec:introduction}

A complex velocity potential can be defined for two-dimensional incompressible potential flows. The complex velocity potential includes, in one function, information about both the real velocity potential and the stream function, and its derivative in respect to the complex position is the conjugate of the velocity, termed complex velocity. The advantages of working with analytic functions in the complex plane made  \hspace{12cm}
\end{fmtext}
\maketitle

\noindent
this representation ubiquitous in the field of theoretical aerodynamics. Borrowing the words of Munk~\cite{munk1925elements}:

\emph{``The advantage of having to do with one function of one variable only is so great, and moreover this function in practical cases becomes so much simpler than any of the functions which it represents, that it pays to get acquainted with this method even if the student has never occupied himself with complex numbers before.''}

This complex representation of potential flows continues to be an integral part of the recent research in fluid mechanics, both for canonical~\cite{krishnamurthy2020transformation,krishnamurthy2021liouville,christopher2021hollow,crowdy2021h,stremler2021something} and applied~\cite{devoria2021theoretical,gehlert2021boundary,saini2021leading,limbourg2021extended,dos2021viscous,semenov2021impulsive} flows. In particular, the complex potential is used in both classical and modern monographs about vortex dynamics~\cite{lamb1932hydrodynamics,saffman1992vortex,alekseenko2007theory}. However, when it comes to obtaining the stability properties of vortex configurations, most authors prefer representations of the position and velocities of the vortices using real vectors. This can be observed in the seminal works of von Kármán~\cite{karman1912mechanismus, karman1912uber} and Art. 156 of Lamb\cite{lamb1932hydrodynamics}, where the stability of the symmetrical and staggered von Kármán vortex streets are obtained. Recent examples of this approach are the references~\cite{mowlavi2016spatio,dynnikova2021stability}. These authors employ complex numbers not for the $y$-axis, but to represent the perturbation, in the form of a complex exponential \emph{ansatz} (as the Fourier modes of a general periodic solution with a growth rate). When the complex velocity is used in the same work as the complex exponential \emph{ansatz}, it is in a different section, as has been done by von Kármán\cite{karman1912mechanismus} himself; both representations are not employed simultaneously.

Milne-Thomson\cite{milne1968theoretical}, on the other hand, employs the complex velocity to obtain the stability of the von Kármán vortex street, by imposing a completely real perturbation, in the form of a cosinusoidal \emph{ansatz} (without the sinusoidal part). However, as opposed to the complex \emph{ansatz}, a cosinusoidal \emph{ansatz} is not a general perturbation and restricts the allowed solutions. Henceforth, the eigenvalues calculated by~\cite{milne1968theoretical} disagree with the eigenvalues of~\cite{karman1912mechanismus, karman1912uber, lamb1932hydrodynamics}. A completely real cosinusoidal perturbation in~\cite{rosenhead1929vii} has also been identified as a cause of a mistake in the study of the stability of a confined vortex street\cite{jimenez1987linear,mowlavi2016spatio}. In extreme cases, the imposition of a perturbation that is not general enough could lead to erroneous conclusions. The year before von Kármán published the now-famous result $\pi h=\cosh^{-1}{\sqrt{2}}$ for the stable configuration~\cite{karman1912mechanismus}, he published a paper with the different result $\pi h=\cosh^{-1}{\sqrt{3}}$~\cite{karman1911mechanismus} ($h$ defined in figure~\ref{fig:vortexstreet}). According to~\cite{meleshko2007bibliography}, there was nothing wrong with von Kármán's math, he just imposed a perturbation on only one vortex pair and kept the other vortices fixed in their position, instead of perturbing the whole vortex street, in fact solving a different problem.

In chapter 7.5 of Saffman\cite{saffman1992vortex}, the stability of a row of vortices is analysed by representing both the complex velocity and the complex exponential \emph{ansatz} using the same set of complex numbers, defined by a single imaginary number $i$. In chapter 7.6 the results of the stability of the von Kármán vortex street are presented but the method is not detailed, at one point it is said that the treatment of Lamb\cite{lamb1932hydrodynamics} is followed and, at another point, that an analysis similar to chapter 7.5 is performed. From our understanding, the method with real vectors for position and velocity is used with the complex exponential \emph{ansatz}. That is not only our understanding of the text, but we also notice that the method of chapter 7.5 of~\cite{saffman1992vortex} contains a formal error: the imaginary number $i$ is used to represent both the $y$-axis and the sinusoidal part of the complex exponential function. This could lead to some unexpected mistakes, such as the sinusoidal component of the perturbation in the $x$-axis being indistinguishable from the cosinusoidal component of the perturbation in the $y$-axis (and vice-versa), which would compromise the values of eigenvalues and eigenvectors. The application of the method of chapter 7.5 to the stability of the von Kármán vortex street leads to incorrect results. It should be noticed, however, that for the case of a row of vortices, this formal error is inconsequential. In section~\ref{sec:stability}\ref{sec:stabilityrow}, we present an explanation of why the results are not affected by this error for this specific configuration.

\begin{figure}[t]
  \centering
  \begin{tikzpicture}
    \def\a{1.5}
    \def\r{\a*0.25}
    \def\h{1.2*\a}
    \def\l{0.3*\a}
    \draw[->] (0,0)--(3.5*\a,0) node[below]{x};
    \draw[->] (0,0)--(0,1.5*\h) node[below left]{y};
    \draw[thick,->] (0,0)--(\l,\h) node[midway,right,black]{$d$};
    \fill (0,0) circle[radius=1pt];
    \draw[-stealth,gray] (0,0) +(150:\r) arc(150:430:\r);
    \fill (\l,\h) circle[radius=1pt];
    \draw[-stealth,gray] (\l,\h) +(170:\r) arc(170:-70:\r);
    \fill (\a,0) circle[radius=1pt];
    \draw[-stealth,gray] (\a,0) +(150:\r) arc(150:430:\r);
    \fill (\l+\a,\h) circle[radius=1pt];
    \draw[-stealth,gray] (\l+\a,\h) +(170:\r) arc(170:-70:\r) node[midway,above,black]{$z_{2n}^0=a n + d$} node[midway,below=20pt,left=2pt,black]{$-\kappa$};
    \fill (2*\a,0) circle[radius=1pt];
    \draw[-stealth,gray] (2*\a,0) +(150:\r) arc(150:430:\r) node[midway,below,black]{$z_{1m}^0=a m$} node[midway,above=22pt,left=0pt,black]{$\kappa$};
    \fill (\l+2*\a,\h) circle[radius=1pt];
    \draw[-stealth,gray] (\l+2*\a,\h) +(170:\r) arc(170:-70:\r);
    \fill (3*\a,0) circle[radius=1pt];
    \draw[-stealth,gray] (3*\a,0) +(150:\r) arc(150:430:\r);
    \fill (\l+3*\a,\h) circle[radius=1pt];
    \draw[-stealth,gray] (\l+3*\a,\h) +(170:\r) arc(170:-70:\r);
    \fill (-\a,0) circle[radius=1pt];
    \draw[-stealth,gray] (-\a,0) +(150:\r) arc(150:430:\r);
    \fill (\l-\a,\h) circle[radius=1pt];
    \draw[-stealth,gray] (\l-\a,\h) +(170:\r) arc(170:-70:\r);
    \fill (-2*\a,0) circle[radius=1pt];
    \draw[-stealth,gray] (-2*\a,0) +(150:\r) arc(150:430:\r);
    \fill (\l-2*\a,\h) circle[radius=1pt];
    \draw[-stealth,gray] (\l-2*\a,\h) +(170:\r) arc(170:-70:\r);
    \draw[decoration={brace,mirror,raise=2pt,amplitude=5pt},decorate]  (0,0) -- (\a,0) node[midway,below=6pt,black]{$a$};
    \draw[decoration={brace,raise=2pt,amplitude=5pt},decorate]  (0,0) -- (0,\h) node[midway,left=6pt,black]{$ah$};
    \draw (-3.0*\a,0)  node{Row 1};
    \draw (-3.0*\a,\h) node{Row 2};
  \end{tikzpicture}
  \caption{Parameters of the von Kármán vortex street.}
  \label{fig:vortexstreet}
\end{figure}
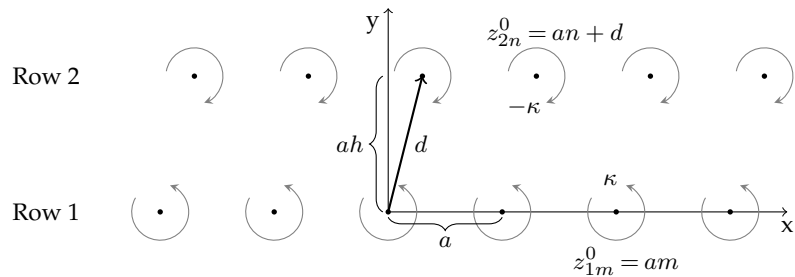

In an effort to employ the two complex representations with the same set of complex numbers in the study of the stability of vortices,~\cite{meiron1984linear,crowdy2002stability} defined ``independent'' variables that can be interpreted as an unusual type of conjugation of other variables. However, such an approach can lead to mistakes, since two distinct interpretations of the complex numbers are being represented with the same complex unit and extreme care should be taken to avoid mixing complex variables.

In order to reconcile the use of both the complex velocity and the complex exponential \emph{ansatz} without the risk of mixing or misinterpreting the terms, we propose the use of bicomplex numbers. In the set of bicomplex numbers, two distinct imaginary units, $i$ and $j$, are used simultaneously. The use of bicomplex numbers instead of the more known extension of the complex number system, the quaternions, is due to the need of maintaining the commutative property. The strategy of~\cite{meiron1984linear,crowdy2002stability} seems to be trying to emulate properties that come naturally when bicomplex numbers are used directly, as explained in section~\ref{sec:framework}\ref{sec:workbc}.

The idea of employing separate imaginary units to represent different quantities in potential flows is not new. Wu~\cite{wu1961swimming} has proposed the representation of space and time variables with different complex units, in the study of the harmonic response of a waving plate, an approach also adopted in other works~\cite{moore2014analytical,moore2017fast,hauge2021new}. In an extension of~\cite{wu1961swimming}, Baddoo~\emph{et al.}\cite{baddoo2021generalization} employed the matrix formulation that is equivalent to imaginary numbers for the complex exponential, in order to avoid two distinct imaginary units. None of these previous works mentioned explicitly the terms bicomplex or multicomplex numbers, so possibly the authors were unfamiliar with the literature concerning these types of numbers.

Nevertheless, becoming aware of the set of bicomplex numbers has the benefit of drawing from its formalism, properties and theorems, in addition to also knowing the limitations of the method. Regarding the algebra of bicomplex numbers, the reader is referred to~\cite{price1991introduction,alpay2014basics,luna2015bicomplex}. In section~\ref{sec:bicomplex}, a brief introduction of bicomplex numbers focusing on the nomenclature used in this paper is presented. In section~\ref{sec:framework}\ref{sec:workbc} we show that the lack of applicability of the Fundamental Theorem of Algebra and the existence of zero divisors in bicomplex algebra are not issues for the current application.

A framework for solving linear stability problems of two-dimensional potential flows using bicomplex numbers is presented in section~\ref{sec:framework}. In section~\ref{sec:stability}\ref{sec:stabilitykarman} this framework is applied to the von Kármán vortex street and a generalized version of the linearized dynamic system is found analytically. The particular cases of the stability of the symmetric and staggered vortex streets, derived from this generalized system, agree perfectly with the classical results from the literature.

As far as the authors are aware, this is the first explicit application of bicomplex algebra to fluid mechanics. The authors were only able to identify implicit applications, mentioned before, or the use of bicomplex numbers in a complex-step differentiation process, such as in~\cite{mons2014isotropic}, where bicomplex numbers are not essential but merely employed to improve the numerical accuracy of derivatives. Most importantly, the bicomplex formalism is a natural step to unify the two uses of complex numbers in stability calculation of potential flows, such that it was previously used~\cite{wu1961swimming,moore2014analytical,moore2017fast,hauge2021new} without the mention to bicomplex numbers. Not only it could have prevented some of the previously mentioned formal errors and incompleteness, but it provides a more convenient, concise and general formulation. As discussed in section~\ref{sec:stability}\ref{sec:stabilitykarman}, the use of bicomplex numbers was essential to arrive at a general system for the von Kármán vortex street, of which the classical results from the literature are particular cases.

Here the framework is applied and validated to the case of a configuration of vortices, which is a problem of interest in many fluid mechanics applications and also from an applied mathematics perspective (see, for example,~\cite{meleshko2007bibliography,alekseenko2007theory,aref2007point}). Nevertheless, this framework is believed to be applicable to all two-dimensional incompressible potential flows, making it possible to apply it to more sophisticated configurations, such as flows represented in the complex plane with the assistance of conformal maps.

\section{Bicomplex numbers} \label{sec:bicomplex}

The set of bicomplex numbers $\mathbb{BC}$ is a higher dimension generalization of the set of complex numbers that maintain the commutative properties. It is a subset of the set of multicomplex numbers. In this section, a brief introduction of bicomplex numbers focusing on the nomenclature used in this paper is presented. For a more comprehensive introduction, the reader is referred to the monographs of Price\cite{price1991introduction}, Alpay \emph{et al.}\cite{alpay2014basics} and Luna-Elizarrar\'{a}s \emph{et al.}\cite{luna2015bicomplex}, on which this section is based.

An element $b$ in the space $\mathbb{BC}$ is a number in the form
\begin{equation}
    b = (x_1 + i y_1) + j (x_2 + i y_2) = x_1 + i y_1 + j x_2 + i j y_2
\end{equation}
where $x_1, x_2, y_1, y_2 \in \mathbb{R}$, and $i$ and $j$ are distinct commuting imaginary units, i.e.,
\begin{equation}
    i^2=j^2=-1, \quad i \neq \pm j, \quad ij=ji .
\end{equation}
Since two imaginary numbers coexist inside $\mathbb{BC}$, the nomenclature $\mathbb{C}(i)$ and $\mathbb{C}(j)$ is used to distinguish between the two sets of complex numbers inside $\mathbb{BC}$. For example, if $z \in \mathbb{C}(i)$, $z$ can be written in the form $z=x+iy$, where $x, y \in \mathbb{R}$. Equivalent definitions of $b$ are 
\begin{equation}
    b = z_1 + j z_2, \quad z_1, z_2 \in \mathbb{C}(i)
\end{equation}
\begin{equation}
    b = \eta_1 + i \eta_2, \quad \eta_1, \eta_2 \in \mathbb{C}(j)
\end{equation}

The following definitions are used for three conjugations for bicomplex numbers, in analogy to the usual complex conjugation:
\begin{enumerate}
    \item $\overline{b} \coloneqq \overline{z_1} + j \overline{z_2} = (x_1 - i y_1) + j (x_2 - i y_2)$ (the bar-conjugation)
    \item $b^{\dagger} \coloneqq z_1 - j z_2 = (x_1 + i y_1) - j (x_2 + i y_2) $ (the $^{\dagger}$-conjugation)
    \item $b^{*} \coloneqq \overline{b}^{\dagger} = \overline{b^{\dagger}} = \overline{z_1} - j \overline{z_2} = (x_1 - i y_1) - j (x_2 - i y_2)$ (the $^{*}$-conjugation)
\end{enumerate}
where $\overline{z_1}, \overline{z_2}$ denotes the usual complex conjugates in $\mathbb{C}(i)$.

In this work, we introduce the notation
\begin{enumerate}
    \item $Re^i(b) \coloneqq (b + \overline{b})/2 = \eta_1 = x_1 + j x_2 \in \mathbb{C}(j)$
    \item $Im^i(b) \coloneqq (b - \overline{b})/(2i) = \eta_2 = y_1 + j y_2 \in \mathbb{C}(j)$
    \item $Re^j(b) \coloneqq (b + b^{\dagger} )/2 = z_1 = x_1 + i y_1 \in \mathbb{C}(i)$
    \item $Im^j(b) \coloneqq (b - b^{\dagger} )/(2j) = z_2 = x_2 + i y_2 \in \mathbb{C}(i)$ .
\end{enumerate}
Thus, $b = Re^i(b) + i Im^i(b) = Re^j(b) + j Im^j(b)$. We recognize that this is an abuse of notation, since $Re^i(b) \notin \mathbb{R}$ (and analogously for the other cases). However, the proposed interpretation is that $Re^i(b)$ is the part of $b$ that do not have a component in the imaginary number $i$, just as the usual definition of $Re(z)$ for $z \in \mathbb{C}(i)$; and $Im^i(b)$ is the part of $b$ that is multiplied by the imaginary number $i$, just as $Im(z)$. We believe that this notation should be more intuitive for people used to complex numbers than a made-up notation that is not drawn from the usual $Re$ and $Im$ notations.

Note that, since $(ij)^2=1$, $\mathbb{BC}$ contains a set of hyperbolic numbers (also called split-complex numbers). We choose not to introduce a new nomenclature for the basic hyperbolic unit, and represent it as $ij$ or $ji$ in this work.

Similarly to hyperbolic numbers, the set of bicomplex numbers also contains zero divisors. Also, the Fundamental Theorem of Algebra is not valid in its original form~\cite{pogorui2006set,luna2015bicomplex}. These facts could pose a challenge for the present application, because the linear stability analysis involves finding eigenvalues, that is performed by finding the roots of polynomials. However, we show in section~\ref{sec:framework}\ref{sec:workbc} that this is not an issue for the present application, because all polynomial coefficients and solutions belong to $\mathbb{C}(j)$.

A possible challenge for the widespread adoption of the method proposed here is the lack of native support for multicomplex numbers on the current programming languages. Purely analytical studies, such as this one, is not affected by this issue. For problems that require the use of software, the user would probably require external packages, such as the one developed in~\cite{casado2020algorithm} for Matlab.


\section{Framework} \label{sec:framework}
In this section, we present the framework to study the stability of potential flows keeping the complex representation of the complex velocity and the complex exponential function. The first steps of the framework are the usual linearization process (see, for example,~\cite{lamb1932hydrodynamics,saffman1992vortex,milne1968theoretical}), that are repeated here for completeness. The main difference in the linearization process is in step~\ref{item:complexansatz}, where a different complex unit $j$ is used to represent the complex exponential. After the linearized dynamics are obtained, three different approaches to study the stability of the system are discussed in the following sections. The steps can be summarized as:

\begin{enumerate}
  \item Define the complex potential $\Phi(z) \in \mathbb{C}(i)$ in terms of the complex position $z=x+iy$. This is usually done by combination of distributed or concentrated sources, sinks, vortices, doublets and other elementary solutions of the Laplace's equation in terms of analytical complex functions (see, for example,~\cite{munk1925elements,lamb1932hydrodynamics,saffman1992vortex,alekseenko2007theory,milne1968theoretical}). \label{item:complexpotential}
  \item Find the complex velocity function $w=u - i v$ (defined as the conjugate of the velocity vector in complex notation) by deriving $\Phi(z)$~\cite{munk1925elements,lamb1932hydrodynamics,saffman1992vortex,alekseenko2007theory,milne1968theoretical}:
  \begin{equation}
    \frac{\partial \overline{z}}{\partial t} = w(z) = \frac{\partial \Phi}{\partial z}
  \end{equation}
  Step~\ref{item:complexpotential} can be skipped if $w(z)$ is directly constructed by elementary solutions.
  \item Find a steady state (or baseflow) around which the equations are going to be linearized, such that
  \begin{equation}
    \frac{\partial \overline{z^0}}{\partial t} = w(z^0)
    \label{eq:steadyw}
  \end{equation}
  does not change in time, where the superscript $^{0}$ indicates undisturbed values. In particular, for flow structures that are advected by the flow, such as vortices, 
  \begin{equation}
    \frac{\partial \overline{z^0}}{\partial t} = w(z^0) = 0 .
  \end{equation}
  \item Define which flow structures (vortices, vortex sheets, sources, etc) are going to be perturbed. Denoting $\alpha$ as the index of each of the $N$ sets of flow structures that are perturbed independently (for example, $\alpha=1$ and $\alpha=2$ for the $N=2$ rows of the von Kármán vortex street), the positions of the flow structures are denoted as $z_{\alpha n}$. In this case, $n$ indicates the index of the flow structure within the elements of set $\alpha$.
  
  For example, $z_{2,4}$ indicates the position of the vortex of index $n=4$ of the second row of the vortex street, that is disturbed simultaneously to the vortices located at $z_{2,0}$, $z_{2,1}$, $z_{2,2}$, etc (see figure~\ref{fig:vortexstreet} and section~\ref{sec:stability}\ref{sec:stabilitykarman} for more detailed definition of the indices and positions of the vortices for this example).
  \item Impose disturbances to the steady state, in the form
  \begin{equation}
    \begin{split}
    z = z^0 + z' \\
    z_{\alpha n} = z^0_{\alpha n} + z_{\alpha n}' .
    \end{split}
  \end{equation}
  \item Employ the Taylor expansion for multiple variables (\emph{w.r.t.} $z'$ and all $z_{\alpha n}'$) to find the linearized equations
  \begin{equation}
    \frac{\partial \overline{z^0}}{\partial t} + \frac{\partial \overline{z'}}{\partial t} = w(z^0) + \frac{\partial w (z^0)}{\partial z} z' + \sum_{\alpha=1}^{N} \, \, \sideset{}{^*} \sum_{n=-\infty}^{+\infty} \frac{\partial w (z^0)}{\partial z_{\alpha n}} z_{\alpha n}' + \mathcal{O}(z'^2,z_{\alpha n}'^2) .
  \end{equation}
  where the symbol $^*$ in the sum indicates that the respective term is avoided if $z=z_{\alpha n}$. Neglecting the second-order terms and using equation~\ref{eq:steadyw}
  \begin{equation}
    \frac{\partial \overline{z'}}{\partial t} = \frac{\partial w (z^0)}{\partial z} z' + \sum_{\alpha=1}^{N} \, \, \sideset{}{^*} \sum_{n=-\infty}^{+\infty} \frac{\partial w (z^0)}{\partial z_{\alpha n}} z_{\alpha n}'.
    \label{eq:framelinearization}
  \end{equation}
  \item Apply the complex exponential \emph{ansatz} (or kernel), defined in terms of the complex unit $j$:
  \begin{equation}
    z_{\alpha n}' = \hat{z_{\alpha}} e^{\varphi n j}
  \end{equation}
  where $\varphi$ can be interpreted as the wavenumber of the disturbance. For continuous structures such as a vortex sheet, a continuous disturbance can be used (for example, using the kernel $e^{\varphi x j}$).
  \label{item:complexansatz}
  \item Evaluate equation~\ref{eq:framelinearization} at positions that give one independent equation for each $\hat{z_{\alpha}}$ (usually evaluated at $n=0$), to arrive at a system
  \begin{equation}
    \frac{\partial}{\partial t} 
    \begin{bmatrix}
      \overline{\hat{z_{1}}} \\
      \overline{\hat{z_{2}}} \\
      \vdots \\
      \overline{\hat{z_{\alpha}}} \\
      \vdots \\
      \overline{\hat{z_{N}}}
    \end{bmatrix}
    = \mathbf{M}
    \begin{bmatrix}
      \hat{z_{1}} \\
      \hat{z_{2}} \\
      \vdots \\
      \hat{z_{\alpha}}
      \vdots \\
      \hat{z_{N}}
    \end{bmatrix}
  \end{equation}
  where $\mathbf{M} \in \mathbb{BC}^{N \times N}$.
  \item Defining
  \begin{equation}
    \mathbf{\hat{z}} = 
    \begin{bmatrix}
      \hat{z_{1}} \\
      \hat{z_{2}} \\
      \vdots \\
      \hat{z_{\alpha}}
      \vdots \\
      \hat{z_{N}}
    \end{bmatrix}
  \end{equation}
  the linear system
  \begin{equation}
    \frac{\partial \overline{\mathbf{\hat{z}}}}{\partial t} = \mathbf{M} \mathbf{\hat{z}}
    \label{eq:framelinearsystem}
  \end{equation}
  that defines the linearized dynamics is found.
  \item From equation~\ref{eq:framelinearsystem}, there are many options to calculate the eigenvalues and eigenvectors that define the linear stability of the flow. The three approaches discussed in sections~\ref{sec:framework}\ref{sec:workcj},~\ref{sec:workbc} and~\ref{sec:worksecond} are, respectively:
  \begin{itemize}
    \item Going back to a single complex plane and working within $\mathbb{C}(j)$, in an approach similar to the traditional method (for example~\cite{karman1912uber,lamb1932hydrodynamics});
    \item Working in the bicomplex space $\mathbb{BC}$;
    \item Working in the bicomplex space $\mathbb{BC}$ with the second-order linear system, in an approach similar to chapter 7.5 of~\cite{saffman1992vortex}.
  \end{itemize}
  The differences and challenges of these three approaches and the connection between them are discussed in sections~\ref{sec:framework}\ref{sec:workcj} to~\ref{sec:worksecond}.
\end{enumerate}

\subsection{First approach: working in $\mathbb{C}(j)$} \label{sec:workcj}
The traditional approach (for example,~\cite{karman1912uber,lamb1932hydrodynamics}) of working with $x$ and $y$ in the real plane can be derived from equation~\ref{eq:framelinearsystem}, by going back to $\mathbb{C}(j)$. This means that we still work in the complex plane defined by $j$ in the \emph{ansatz}, but with positions of the vortices, $\mathbf{\hat{x}}=Re^i(\mathbf{\hat{z}})$ and $\mathbf{\hat{y}}=Im^i(\mathbf{\hat{z}})$, that do not have a component in $i$.

Applying $Re^i$ and $Im^i$ to equation~\ref{eq:framelinearsystem}:
\begin{equation}
  \frac{\partial }{\partial t} \begin{bmatrix}
    \mathbf{\hat{x}} \\
    \mathbf{\hat{y}}
  \end{bmatrix}
  =
  \begin{bmatrix}
    Re^i \left(\frac{\partial \overline{\mathbf{\hat{z}}}}{\partial t} \right) \\
    -Im^i \left(\frac{\partial \overline{\mathbf{\hat{z}}}}{\partial t} \right)
  \end{bmatrix}
   =
  \begin{bmatrix}
    Re^i(\mathbf{M}) & -Im^i(\mathbf{M}) \\
    -Im^i(\mathbf{M}) & -Re^i(\mathbf{M}) 
  \end{bmatrix}
  \begin{bmatrix}
    \mathbf{\hat{x}} \\
    \mathbf{\hat{y}}
  \end{bmatrix} .
\end{equation}

Defining 
\begin{equation}
  \mathbf{R} = \begin{bmatrix}
    Re^i(\mathbf{M}) & -Im^i(\mathbf{M}) \\
    -Im^i(\mathbf{M}) & -Re^i(\mathbf{M}) 
  \end{bmatrix}
\end{equation}
the stability properties of the system are given by the eigenvalues and eigenvectors of the matrix $\mathbf{R} \in \mathbb{C}(j)^{2N \times 2N}$, where $2N$ is the size of the system, obtained by the number of complex variables $N$ (in the case of von Kármán vortex street $N=2$) multiplied by the dimension of the problem, $2$ ($x$ and $y$ directions).

In order to study the linear stability of the system, solutions formed by complex exponential functions are assumed. The substitution
\begin{equation}
  \begin{bmatrix}
    \mathbf{\hat{x}} \\
    \mathbf{\hat{y}}
  \end{bmatrix}
  =
  \begin{bmatrix}
    \mathbf{\hat{\hat{x}}} \\
    \mathbf{\hat{\hat{y}}}
  \end{bmatrix} e^{\lambda t} ,
  \label{eq:tempansatz}
\end{equation}
with $\lambda$ as the temporal growth rate of disturbances, gives the eigenvalue problem
\begin{equation}
  \mathbf{R} \begin{bmatrix}
    \mathbf{\hat{\hat{x}}} \\
    \mathbf{\hat{\hat{y}}}
  \end{bmatrix}
  =
  \lambda \begin{bmatrix}
    \mathbf{\hat{\hat{x}}} \\
    \mathbf{\hat{\hat{y}}}
  \end{bmatrix}
\end{equation}
where it is clear that $\lambda \in \mathbb{C}(j)$ is an eigenvalue of $\mathbf{R}$. The eigenvalues and eigenvectors of $\mathbf{R}$ can be complex. In this case, to be physically meaningful, the eigenvalues and eigenvectors should lie in $\mathbb{C}(j)$. The imaginary part of the eigenvalues is the frequency of the response and the imaginary part of the eigenvectors is related to the phase of the complex exponential function. Hence, the eigenvalue decomposition can be performed as usual, just remembering to use the imaginary unit $j$ if needed (in this context, $i$ has no meaning).


This approach is used to compare our results with the results from previous studies~\cite{karman1912uber,lamb1932hydrodynamics} in section~\ref{sec:stability}.

\subsection{Second approach: working in $\mathbb{BC}$} \label{sec:workbc}
To arrive at an eigenvalue problem using $\mathbf{\hat{z}}$ and without requiring $\mathbf{\hat{x}}$ or $\mathbf{\hat{y}}$, we take the bar-conjugate of equation~\ref{eq:framelinearsystem} (conjugate in respect to $i$):
\begin{equation}
  \frac{\partial \mathbf{\hat{z}}}{\partial t} = \overline{\mathbf{M}} \, \overline{\mathbf{\hat{z}}}
  \label{eq:Azconj}
\end{equation}
hence
\begin{equation}
  \frac{\partial }{\partial t} \begin{bmatrix}
    \mathbf{\hat{z}} \\
    \overline{\mathbf{\hat{z}}}
  \end{bmatrix}
   =
  \begin{bmatrix}
    0            & \overline{\mathbf{M}} \\
    \mathbf{M} & 0
  \end{bmatrix}
  \begin{bmatrix}
    \mathbf{\hat{z}} \\
    \overline{\mathbf{\hat{z}}}
  \end{bmatrix}
  \label{eq:Qsystem}
\end{equation}
which is a linear system of the same size as the one described in section~\ref{sec:workcj}, $2N$. Defining
\begin{equation}
  \mathbf{Q} =
  \begin{bmatrix}
    0            & \overline{\mathbf{M}} \\
    \mathbf{M} & 0
  \end{bmatrix}
\end{equation}
it is easy to see that $\mathbf{Q} \in \mathbb{BC}^{2N \times 2N}$.

The temporal response is studied by assuming the same form of equation~\ref{eq:tempansatz}:
\begin{equation}
    \mathbf{\hat{z}}=\mathbf{\hat{x}} + i \mathbf{\hat{y}} = (\mathbf{\hat{\hat{x}}} + i \mathbf{\hat{\hat{y}}}) e^{\lambda t} = \mathbf{\hat{\hat{z}}} e^{\lambda t}
\end{equation}
where $\lambda \in \mathbb{C}(j)$. Possible solutions for $\lambda \in \mathbb{BC}$ that have an imaginary component in $\mathbb{C}(i)$ are not suitable to study the time evolution of a perturbation. The multiplication by $e^{\lambda t}$ for $\lambda$ with non-zero $i$-part would indicate a rotation in the $x$-$y$-plane, which does not correspond to a growth of the eigenvectors. Since 
\begin{equation}
    \overline{\mathbf{\hat{z}}} = \overline{\mathbf{\hat{\hat{z}}}} e^{\lambda t}
    \label{eq:conjtempansatz}
\end{equation}
is valid for $\lambda \in \mathbb{C}(j)$, we arrive at the eigenvalue problem:
\begin{equation}
  \mathbf{Q} \begin{bmatrix}
    \mathbf{\hat{\hat{z}}} \\
    \overline{\mathbf{\hat{\hat{z}}}}
  \end{bmatrix}
  =
  \lambda \begin{bmatrix}
    \mathbf{\hat{\hat{z}}} \\
    \overline{\mathbf{\hat{\hat{z}}}}
  \end{bmatrix} .
  \label{eq:eigenproblemQ}
\end{equation}
The derivation of equation~\ref{eq:eigenproblemQ} also corroborates that solutions of $\lambda$ with non-zero part in $\mathbb{C}(i)$ do not have a physical meaning: in equation~\ref{eq:conjtempansatz}, the property that $\overline{\lambda} = \lambda$ for $\lambda \in \mathbb{C}(j)$ is used. 

From this, it is possible to interpret the ``independent'' variables of~\cite{meiron1984linear,crowdy2002stability} as the bar-conjugate of the main variables, similar to how the equations related to $\overline{\mathbf{\hat{z}}}$ are needed here (equation~\ref{eq:Azconj}). In those works, taking the conjugate of a variable would be equivalent to taking the $^{*}$-conjugate in the bicomplex formulation, which would change the sign of the imaginary part of $\lambda$, among other potential problems, what motivated their unusual approach. In the bicomplex formulation, it is straightforward to notice that the extra equations required can be obtained directly from the bar-conjugate of the original equations.

To find the eigenvalues of $\mathbf{Q}$, the roots of its characteristic polynomial should be found. Finding the roots of a polynomial in bicomplex algebra poses some challenges, because the Fundamental Theorem of Algebra is not valid ~\cite{pogorui2006set,luna2015bicomplex}. This means that it is not possible to guarantee that a polynomial of order $2N$ will have $2N$ roots, as it is for complex numbers (when considering multiplicity). However, for this matrix, it is not necessary to rely on the analogue of the Fundamental Theorem of Algebra for bicomplex polynomials~\cite{luna2015bicomplex}, because the coefficients of the characteristic polynomials lie in $\mathbb{C}(j)$ and we are interested in roots that lie in $\mathbb{C}(j)$.

\begin{theorem} \label{the:polynomialQ}
  If a matrix $\mathbf{Q} \in \mathbb{BC}^{2N \times 2N}$ can be written as
  \begin{equation}
    \mathbf{Q} =
    \begin{bmatrix}
      0            & \overline{\mathbf{M}} \\
      \mathbf{M} & 0
    \end{bmatrix}
  \end{equation}
  where $\mathbf{M} \in \mathbb{BC}^{N \times N}$ then all the coefficients of the characteristic equation of $\mathbf{Q}$ lie in $\mathbb{C}(j)$.
\end{theorem}

\begin{proof}
  According to~\cite{pipes2014applied} (section 3.21), the coefficients of the characteristic polynomial of a matrix can be found using a recursive formula based on the trace of matrices. The coefficients $c_1, c_2, \dots, c_n$ of the characteristic equation $\lambda^n + c_1 \lambda^{n-1} + \dots + c_{n-1} \lambda  + c_n = 0$ are
  \begin{equation}
    c_n = - \frac{1}{n} (c_{n-1} s_{1} + c_{n-2} s_{2} + \dots + c_{1} s_{n})
  \end{equation}
  where $s_n = Tr(\mathbf{Q}^n)$, $Tr()$ denoting the trace of the matrix.
  
  By induction is easy to show that
  \begin{equation}
    \mathbf{Q}^{2q+1} =
    \begin{bmatrix}
      0            & (\overline{\mathbf{M}}\mathbf{M})^{q} \overline{\mathbf{M}} \\
      (\mathbf{M}\overline{\mathbf{M}})^{q} \mathbf{M} & 0
    \end{bmatrix}
  \end{equation}
  and 
  \begin{equation}
    \mathbf{Q}^{2q} =
    \begin{bmatrix}
      (\overline{\mathbf{M}}\mathbf{M})^{q}  & 0 \\
      0 &  (\mathbf{M}\overline{\mathbf{M}})^{q}
    \end{bmatrix}
  \end{equation}
  for every positive integer $q$. Hence, for odd $n=2q+1$ it follows directly that $s_{2q+1} = 0$. For even $n=2q$:
  \begin{equation}
  \begin{split}
    s_{2q} & = Tr(\mathbf{Q}^{2q}) = Tr((\overline{\mathbf{M}}\mathbf{M})^{q}) + Tr((\mathbf{M}\overline{\mathbf{M}})^{q}) \\
    & = Tr(Re^i((\overline{\mathbf{M}}\mathbf{M})^{q}) + i Im^i((\overline{\mathbf{M}}\mathbf{M})^{q})) + Tr(Re^i((\mathbf{M}\overline{\mathbf{M}})^{q}) + i Im^i((\mathbf{M}\overline{\mathbf{M}})^{q})) \\
    & = Tr(Re^i((\overline{\mathbf{M}}\mathbf{M})^{q}+ i Im^i((\overline{\mathbf{M}}\mathbf{M})^{q})) + Tr(Re^i((\overline{\mathbf{M}}\mathbf{M})^{q}) - i Im^i((\overline{\mathbf{M}}\mathbf{M})^{q})) \\
    & = 2 Tr(Re^i((\overline{\mathbf{M}}\mathbf{M})^{q})) \in \mathbb{C}(j) .
  \end{split}
  \end{equation}
  
  Because $s_{n}$ lie in $\mathbb{C}(j)$ for every $n$, the proof that $c_n$ lie in $\mathbb{C}(j)$ is straightforward. For $n=1$, $c_{1}=-s_{1}=0 \in \mathbb{C}(j)$. Since $c_n$ is defined by a recursive formula, we conclude by induction that $c_n \in \mathbb{C}(j)$ for every $n$. 
\end{proof}

Therefore, the characteristic equation of the matrix $\mathbf{Q}$ lies in $\mathbb{C}(j)$ and we are interested in solutions in $\mathbb{C}(j)$, which are physically meaningful. Eigenvalues that have a component in $\mathbb{C}(i)$ are not physically meaningful and are not mathematically consistent with the derivation of equation~\ref{eq:eigenproblemQ}, because they do not satisfy equation~\ref{eq:conjtempansatz}. It should be noticed that the characteristic equation may have other solutions that do not lie in $\mathbb{C}(j)$ if solved in $\mathbb{BC}$, even if the coefficients all lie in $\mathbb{C}(j)$. However, these other solutions are not of interest to our application. For example, $\lambda^2+1=0$ have the solutions $\lambda=\pm j$ that are of interest and also the solutions $\lambda=\pm i$, that are not of interest.

Since the coefficients of the characteristic equation lie in $\mathbb{C}(j)$ it is easy to restrict ourselves to the solutions of interest and the eigenvalue problem can be solved by relying on the Fundamental Theorem of Algebra for complex numbers. Hence, there are $2N$ eigenvalues that belong to $\mathbb{C}(j)$ (when considering multiplicity), as expected according to the representation of section~\ref{sec:workcj}.

\begin{remark} \label{rem:eigenvaluesQ}
  All the $2N$ eigenvalues of interest (when considering multiplicity) of the matrix $\mathbf{Q}$ (defined in theorem~\ref{the:polynomialQ}) belong to $\mathbb{C}(j)$ and can be found by solving the characteristic equation of $\mathbf{Q}$ in $\mathbb{C}(j)$.
\end{remark}

Physically the eigenvalues of $\mathbf{R}$ and $\mathbf{Q}$ should be the same. This can also be proved mathematically.

\begin{theorem} \label{the:eigenvaluesRQ}
  If a matrix $\mathbf{Q} \in \mathbb{BC}^{2N \times 2N}$ can be written as
  \begin{equation}
    \mathbf{Q} =
    \begin{bmatrix}
      0            & \overline{\mathbf{M}} \\
      \mathbf{M} & 0
    \end{bmatrix}
  \end{equation}
  and a matrix $\mathbf{R} \in \mathbb{C}(j)^{2N \times 2N}$ can be written as
  \begin{equation}
    \mathbf{R} = \begin{bmatrix}
      Re^i(\mathbf{M}) & -Im^i(\mathbf{M}) \\
      -Im^i(\mathbf{M}) & -Re^i(\mathbf{M}) 
    \end{bmatrix}
  \end{equation}
  where $\mathbf{M} \in \mathbb{BC}^{N \times N}$ then all the eigenvalues in $\mathbb{C}(j)$ of $\mathbf{R}$ and $\mathbf{Q}$ are the same.
\end{theorem}

\begin{proof}
  To prove theorem~\ref{the:eigenvaluesRQ}, we make use of the matrices $\mathbf{F}$ and $\mathbf{F}^{-1}$ used for conversion between the $\mathbf{\hat{z}}, \overline{\mathbf{\hat{z}}}$-system to the $\mathbf{\hat{x}}, \mathbf{\hat{y}}$-system:
  \begin{equation}
    \begin{bmatrix}
      \mathbf{\hat{z}} \\
      \overline{\mathbf{\hat{z}}}
    \end{bmatrix} 
    =
    \begin{bmatrix}
      \mathbf{I} &  i \mathbf{I} \\
      \mathbf{I} & -i \mathbf{I}
    \end{bmatrix}
    \begin{bmatrix}
      \mathbf{\hat{x}} \\
      \mathbf{\hat{y}}
    \end{bmatrix}
    =
    \mathbf{F}
    \begin{bmatrix}
      \mathbf{\hat{x}} \\
      \mathbf{\hat{y}}
    \end{bmatrix}
  \end{equation}
  and 
  \begin{equation}
    \begin{bmatrix}
      \mathbf{\hat{x}} \\
      \mathbf{\hat{y}}
    \end{bmatrix}
    =
    \begin{bmatrix}
      \frac{1}{2}\mathbf{I} & \frac{1}{2}\mathbf{I} \\
      -\frac{i}{2}\mathbf{I} & \frac{i}{2}\mathbf{I}
    \end{bmatrix}
    \begin{bmatrix}
      \mathbf{\hat{z}} \\
      \overline{\mathbf{\hat{z}}}
    \end{bmatrix}
    =
    \mathbf{F}^{-1}
    \begin{bmatrix}
      \mathbf{\hat{z}} \\
      \overline{\mathbf{\hat{z}}}
    \end{bmatrix}
  \end{equation}
  where $\mathbf{I} \in \mathbb{R}^{N \times N}$ is the identity matrix. It is easy to verify that
  \begin{equation}
    \mathbf{Q} = \mathbf{F} \mathbf{R} \mathbf{F}^{-1} .
  \end{equation}
  If the Jordan canonical form of $\mathbf{R}$ is given by $\mathbf{R} = \mathbf{V} \boldsymbol{\Lambda} \mathbf{V}^{-1}$, then the Jordan canonical form of $\mathbf{Q}$ is
  \begin{equation}
    \mathbf{Q} = \mathbf{F} \mathbf{V} \boldsymbol{\Lambda} \mathbf{V}^{-1} \mathbf{F}^{-1} = \mathbf{P} \boldsymbol{\Lambda} \mathbf{P}^{-1}.
  \end{equation}
  This proves that every eigenvalues of $\mathbf{R}$ is an eigenvalue of $\mathbf{Q}$ (with the same multiplicity). Since matrix $\mathbf{Q}$ has $2N$ eigenvalues when solving the characteristic equation in $\mathbb{C}(j)$ (remark~\ref{rem:eigenvaluesQ}), all the eigenvalues of $\mathbf{Q}$ that belong to $\mathbb{C}(j)$ are the same as $\mathbf{R}$ (with the same multiplicity).
\end{proof}

\begin{corollary} \label{cor:eigenvectorsRQ}
  If the matrices $\mathbf{R}$ and $\mathbf{Q}$ of theorem~\ref{the:eigenvaluesRQ} are diagonalizable, then the eigenvectors of $\mathbf{Q}$ are given by $\mathbf{P}=\mathbf{F} \mathbf{V}$, where $\mathbf{V}$ is the matrix of eigenvectors of the $\mathbf{R}$ and
  \begin{equation}
    \mathbf{F}
    =
    \begin{bmatrix}
      \mathbf{I} &  i \mathbf{I} \\
      \mathbf{I} & -i \mathbf{I}
    \end{bmatrix} .
  \end{equation} 
\end{corollary}

\begin{proof}
  It follows directly from the proof of theorem~\ref{the:eigenvaluesRQ}.
\end{proof}

The eigenvectors of $\mathbf{R}$ lie in $\mathbb{C}(j)$. From corollary~\ref{cor:eigenvectorsRQ} we note that the eigenvectors of $\mathbf{Q}$ lie in $\mathbb{BC}$. To find the eigenvectors of $\mathbf{Q}$ we then have to solve a degenerated linear system in $\mathbb{BC}$. Since a solution is known to exist from theorem~\ref{cor:eigenvectorsRQ}, the loss of the Fundamental Theorem of Algebra is not an issue, and obtaining the eigenvectors should be straightforward.

\subsection{Third approach: second-order system} \label{sec:worksecond}
In analogy to the method in chapter 7.5 of Saffman\cite{saffman1992vortex}, the second derivative of $\mathbf{\hat{z}}$ can be derived from equations~\ref{eq:framelinearsystem} and~\ref{eq:Azconj} to find the second-order linear system
\begin{equation}
  \frac{\partial^2 \mathbf{\hat{z}}}{\partial t^2} = \overline{\mathbf{M}} \, \frac{\partial \overline{\mathbf{\hat{z}}}}{\partial t} = \overline{\mathbf{M}} \mathbf{M} \, \mathbf{\hat{z}}
\end{equation}

The procedure of~\cite{saffman1992vortex} is equivalent to using the square root of the eigenvalues of $\overline{\mathbf{M}} \mathbf{M}$ to study the linear system. In theorem~\ref{the:eigMconjM} we show that the eigenvalues of $\overline{\mathbf{M}} \mathbf{M}$ are the square of the eigenvalues of $\mathbf{Q}$.

\begin{theorem} \label{the:eigMconjM}
  For a matrix $\mathbf{Q} \in \mathbb{BC}^{2N \times 2N}$ that can be written as
  \begin{equation}
    \mathbf{Q} =
    \begin{bmatrix}
      0            & \overline{\mathbf{M}} \\
      \mathbf{M} & 0
    \end{bmatrix} ,
  \end{equation}
  where $\mathbf{M} \in \mathbb{BC}^{N \times N}$, the eigenvalues of $\mathbf{Q}$ come in pairs in the form $\pm\lambda_1, \pm\lambda_2, \dots, \pm\lambda_N$ and the eigenvalues of $\overline{\mathbf{M}} \mathbf{M}$ are $\lambda_1^2, \lambda_2^2, \dots, \lambda_N^2$.
\end{theorem}

\begin{proof}
  The eigenvalues of $\mathbf{Q}$ are given by $\lambda$ that solves the equation
  \begin{equation}
    \det
    \begin{bmatrix}
      - \lambda \mathbf{I} & \overline{\mathbf{M}} \\
      \mathbf{M}           & - \lambda \mathbf{I} 
    \end{bmatrix} = 0 \iff
    \det (\lambda^2 \mathbf{I} - \overline{\mathbf{M}} \mathbf{M}) = 0
  \end{equation}
  where the fact that $\mathbf{I}$ and $\mathbf{M}$ commute and the properties of the determinant of block matrices~\cite{silvester2000determinants} were used. This shows that $-\lambda$ is a eigenvalue of $\mathbf{Q}$ if $\lambda$ is a eigenvalue. Also, it follows directly that $\lambda^2$ is an eigenvalue of $\overline{\mathbf{M}} \mathbf{M}$.
\end{proof}

The eigendecomposition of the matrix $\overline{\mathbf{M}} \mathbf{M}$ gives the information about the eigenvalues of $\mathbf{Q}$ and $\mathbf{R}$. However, it can not give complete information about the eigenvectors, since its dimension is $N$, smaller than the full linear system. A second-order linear system can be obtained from equation~\ref{eq:Qsystem}
\begin{equation}
  \frac{\partial^2 }{\partial t^2} \begin{bmatrix}
    \mathbf{\hat{z}} \\
    \overline{\mathbf{\hat{z}}}
  \end{bmatrix}
  =
  \mathbf{Q}^2
  \begin{bmatrix}
    \mathbf{\hat{z}} \\
    \overline{\mathbf{\hat{z}}}
  \end{bmatrix}
  =
  \begin{bmatrix}
    \overline{\mathbf{M}} \mathbf{M} & 0 \\
    0 & \mathbf{M} \overline{\mathbf{M}}
  \end{bmatrix}
  \begin{bmatrix}
    \mathbf{\hat{z}} \\
    \overline{\mathbf{\hat{z}}}
  \end{bmatrix} .
\end{equation}

The matrix $\mathbf{Q}^2$ can also be used to find the eigenvalues of $\mathbf{Q}$, since the eigenvalues $\mu$ of $\mathbf{Q}^2$ are the same as the combination of the eigenvalues of $\overline{\mathbf{M}} \mathbf{M}$ and $\mathbf{M} \overline{\mathbf{M}}$ because
\begin{equation}
  \det \left(\mathbf{Q}^2 - \mu \mathbf{I}^{2N \times 2N} \right) = \det \left(\mathbf{M} \overline{\mathbf{M}} - \mu \mathbf{I}^{N \times N}\right) \det \left(\overline{\mathbf{M}} \mathbf{M} - \mu \mathbf{I}^{N \times N}\right)
\end{equation}
where the properties of the determinant of block matrices were used. However, the eigenvalues of $\mathbf{M} \overline{\mathbf{M}}$ are the same as the eigenvalues of $\overline{\mathbf{M}} \mathbf{M}$, as can be shown by
\begin{equation}
  \det (\mathbf{M} \overline{\mathbf{M}} - \mu \mathbf{I}) = 0 \iff \overline{\det (\mathbf{M} \overline{\mathbf{M}} - \mu \mathbf{I})} = 0 \iff \det (\overline{\mathbf{M}} \mathbf{M} - \mu \mathbf{I}) = 0
\end{equation}
where the property that $\mu=\lambda^2 \in \mathbb{C}(j)$ was used. Hence, the eigenvalues of $\mathbf{Q}^2$ are $\mu_1=\lambda_1^2, \mu_2=\lambda_2^2, \dots, \mu_N=\lambda_N^2$, with multiplicity 2 (assuming $\lambda_p \neq \lambda_q$ for $p \neq q$). Thus, the eigenvalues of $\mathbf{Q}^2$ are not unique and the diagonalization of $\mathbf{Q}^2$ is not unique (if it were unique, the eigenvectors of $\mathbf{Q}^2$ would be the same of $\mathbf{Q}$, since $\mathbf{Q}^2=\mathbf{P} \boldsymbol{\Lambda}^2 \mathbf{P}^{-1}$). Therefore, the eigenvalues of $\mathbf{Q}^2$ can be used to find the eigenvalues of $\mathbf{Q}$ but the eigenvectors of $\mathbf{Q}^2$ should not be used to calculate the eigenvectors of $\mathbf{Q}$.

\section{Stability of classical configurations} \label{sec:stability}

\subsection{Stability of the von Kármán vortex street} \label{sec:stabilitykarman}

The von Kármán vortex street, shown in figure~\ref{fig:vortexstreet} is defined by the distance $a$ between the vortices of the same row and the complex distance $d \in \mathbb{C}(i)$ between vortices in different rows. For the symmetric von Kármán vortex street, $d=i ah$, where $h \in \mathbb{R}$ is the normalized distance in the $y$-direction. For the staggered von Kármán vortex street, $d=(1/2 + i h)a$. The complex position of each undisturbed vortex is:
\begin{equation}
  \begin{split}
    z_{1m}^{0} = a m \\
    z_{2n}^{0} = a n + d
  \end{split}
\end{equation}
where $z_{1m}^{0}, z_{2n}^{0} \in \mathbb{C}(i)$, the superscript $^{0}$ indicates undisturbed values and $m$, $n$ are the indices of the vortices within each row. The circulation of each vortex is
\begin{equation}
  \begin{split}
    \kappa_{1m} = \kappa \\
    \kappa_{2n} = -\kappa
  \end{split}
\end{equation}
where $\kappa \in \mathbb{R}$ for the configuration studied in this paper. In the general case, $\kappa \in \mathbb{C}(i)$ can be used to study configurations including vortices, sinks and sources.

The complex potential at a position $z$ is
\begin{equation}
  \Phi(z) =  - \sideset{}{^*} \sum_{m=-\infty}^{+\infty} \frac{i \kappa_{1m}}{2 \pi} \log(z-z_{1m}) - \sideset{}{^*} \sum_{n=-\infty}^{+\infty} \frac{i \kappa_{2n}}{2 \pi} \log(z-z_{2n})
\end{equation}
where the symbol $^*$ in the sum indicates that the term $\frac{i \kappa_{1m}}{2 \pi} \log(z-z_{1m})$ is avoided if $z=z_{1m}$ and the term $\frac{i \kappa_{2n}}{2 \pi} \log(z-z_{2n})$ is avoided for $z=z_{2n}$ (only one of these terms can be avoided). The complex velocity $w=u - i v$ is given by
\begin{equation}
  \frac{\partial \overline{z}}{\partial t} = w = \frac{\partial \Phi}{\partial z} = - \sideset{}{^*} \sum_{m=-\infty}^{+\infty} \frac{i \kappa_{1m}}{2 \pi} \frac{1}{z-z_{1m}} - \sideset{}{^*} \sum_{n=-\infty}^{+\infty} \frac{i \kappa_{2n}}{2 \pi} \frac{1}{z-z_{2n}}
  \label{eq:completew}
\end{equation}
imposing a disturbance to the position of the vortices:
\begin{equation}
  \begin{split}
    z = z^{0} + z' \\
    z_{1m} = z_{1m}^{0} + z_{1m}' \\
    z_{2n} = z_{2n}^{0} + z_{2n}'
  \end{split}
\end{equation}
and linearizing equation~\ref{eq:completew}, the complex velocity is given by
\begin{equation}
  \begin{split}
  \frac{\partial \overline{z}}{\partial t} = \frac{\partial \overline{z^{0}}}{\partial t}  + \frac{\partial \overline{z'}}{\partial t} \approx & \left( - \sideset{}{^*} \sum_{m=-\infty}^{+\infty}  \frac{i \kappa_{1m}}{2 \pi} \frac{1}{z^{0}-z_{1m}^{0}} - \sideset{}{^*} \sum_{n=-\infty}^{+\infty}  \frac{i \kappa_{2n}}{2 \pi} \frac{1}{z^{0}-z_{2n}^{0}}  \right) \\
  & + \left( \sideset{}{^*} \sum_{m=-\infty}^{+\infty} \frac{i \kappa_{1m}}{2 \pi} \frac{z'-z_{1m}'}{(z^{0}-z_{1m}^{0})^2} + \sideset{}{^*} \sum_{n=-\infty}^{+\infty} \frac{i \kappa_{2n}}{2 \pi} \frac{z'-z_{2n}'}{(z^{0}-z_{2n}^{0})^2}  \right) .
  \end{split}
\end{equation}
If the undisturbed configuration is known to be a steady solution:
\begin{equation}
  \frac{\partial \overline{z^{0}}}{\partial t} = - \sideset{}{^*} \sum_{m=-\infty}^{+\infty}  \frac{i \kappa_{1m}}{2 \pi} \frac{1}{z^{0}-z_{1m}^{0}} - \sideset{}{^*} \sum_{n=-\infty}^{+\infty}  \frac{i \kappa_{2n}}{2 \pi} \frac{1}{z^{0}-z_{2n}^{0}} = 0
\end{equation}
for all the vortices. Hence, the linearized version of equation~\ref{eq:completew} becomes
\begin{equation}
  \frac{\partial \overline{z'}}{\partial t} = \sideset{}{^*} \sum_{m=-\infty}^{+\infty} \frac{i \kappa_{1m}}{2 \pi} \frac{z'-z_{1m}'}{(z^{0}-z_{1m}^{0})^2} + \sideset{}{^*} \sum_{n=-\infty}^{+\infty} \frac{i \kappa_{2n}}{2 \pi} \frac{z'-z_{2n}'}{(z^{0}-z_{2n}^{0})^2}  .
  \label{eq:linearw}
\end{equation}
Which shows that the details of the steady solution do not need to be explicitly known for the stability problem. Only the knowledge of the position of the vortices and the fact that it is a steady configuration is sufficient.

A complex exponential function can be imposed as disturbance without loss of generality, because these are the components of the Fourier series in exponential form of a general disturbance. However, if the complex exponential \emph{ansatz} lies in $\mathbb{C}(i)$, we would be using the complex unit $i$ to represent both the sinusoidal component of the perturbation and the $y$-direction of the complex velocity and positions. To avoid, this, we introduce the other complex unit $j$, defining the perturbation so that the complex exponential lie in $\mathbb{C}(j)$:
\begin{equation}
  \begin{split}
  z_{1m}' = \hat{z_{1}} e^{2 \pi p m j} \\
  z_{2n}' = \hat{z_{2}} e^{2 \pi p n j} .
  \label{eq:ansatz}
  \end{split}
\end{equation}
The $x$ and $y$ components of $\hat{z_{1}}$ are:
\begin{itemize}
    \item $Re^i(\hat{z_{1}}) = \hat{x_{1}} \in \mathbb{C}(j)$
    \item $Im^i(\hat{z_{1}}) = \hat{y_{1}} \in \mathbb{C}(j)$ .
\end{itemize}
Hence, the expanded form of $z_{1m}'$ is:
\begin{equation}
  z_{1m}' = \hat{x_{1}} \cos{(2 \pi p m)} + i \hat{y_{1}} \cos{(2 \pi p m)} + j \hat{x_{1}} \sin{(2 \pi p m)} + ij \hat{y_{1}} \sin{(2 \pi p m)}
\end{equation}
and analogously for $z_{2n}'$. It should be remembered that $\hat{x_{1}}$ and $\hat{y_{1}}$ could have a component in $j$, if there is a phase difference between $\hat{x_{1}}$ and $\hat{y_{1}}$ or between $\hat{z_{1}}$ and $\hat{z_{2}}$.

Substituting the \emph{ansatz} of equation~\ref{eq:ansatz}, the position of the undisturbed vortices and the circulation, equation~\ref{eq:linearw} for vortices with indices $m=n=0$ (without loss of generality) becomes 
\begin{equation}
  \frac{\partial \overline{\hat{z_{1}}}}{\partial t} = \frac{i \kappa}{2 \pi} \left( \sideset{}{'} \sum_{m=-\infty}^{+\infty} \frac{\hat{z_{1}}-\hat{z_{1}} e^{2 \pi p m j}}{(a m)^2} - \sum_{n=-\infty}^{+\infty} \frac{\hat{z_{1}}-\hat{z_{2}} e^{2 \pi p n j}}{(a n + d)^2}  \right)
\end{equation}
\begin{equation}
  \frac{\partial \overline{\hat{z_{2}}}}{\partial t} = \frac{i \kappa}{2 \pi} \left( \sum_{m=-\infty}^{+\infty} \frac{\hat{z_{2}}-\hat{z_{1}} e^{2 \pi p m j}}{(a m-d)^2} - \sideset{}{'} \sum_{n=-\infty}^{+\infty}  \frac{\hat{z_{2}}-\hat{z_{2}} e^{2 \pi p n j}}{(a n)^2}  \right) ,
\end{equation}
where the symbol $'$ in the sum indicates that the term $m,n=0$ is avoided. Defining:
\begin{equation}
  k=\frac{d}{a}
\end{equation}
we have
\begin{equation}
  \frac{\partial \overline{\hat{z_{1}}}}{\partial t} = \frac{i \kappa}{2 a^2 \pi} \left( \left[ \sideset{}{'} \sum_{m=-\infty}^{+\infty} \frac{1-e^{2 \pi p m j}}{m^2}  - \sum_{n=-\infty}^{+\infty} \frac{1}{(n + k)^2} \right] \hat{z_{1}} + \sum_{n=-\infty}^{+\infty} \frac{e^{2 \pi p n j}}{(n + k)^2}\hat{z_{2}}  \right)
\end{equation}
\begin{equation}
  \frac{\partial \overline{\hat{z_{2}}}}{\partial t} = \frac{i \kappa}{2 a^2 \pi} \left( \sum_{m=-\infty}^{+\infty}  \frac{- e^{2 \pi p m j}}{(m-k)^2} \hat{z_{1}}  + \left[ \sum_{m=-\infty}^{+\infty}  \frac{1}{(m-k)^2} - \sideset{}{'} \sum_{n=-\infty}^{+\infty}  \frac{1-e^{2 \pi p n j}}{n^2} \right] \hat{z_{2}}  \right) .
\end{equation}

Using the following identities for $0 \le p \le 1$ (that can derived from the identities presented in section 156 of~\cite{lamb1932hydrodynamics}):
\begin{equation}
   \sideset{}{'} \sum_{n=-\infty}^{+\infty} \frac{1-e^{2 \pi p n j}}{n^2} = 2 \pi^2 p (1 - p)
\end{equation}
\begin{equation}
   \sum_{n=-\infty}^{+\infty} \frac{1}{(n + k)^2} = \frac{\pi^2}{\sin^2(\pi k)}
\end{equation}
\begin{equation}
   \sum_{n=-\infty}^{+\infty} \frac{e^{2 \pi p n j}}{(n + k)^2} = \frac{\pi^2 e^{-2 p \pi k j}}{\sin^2(\pi k)} + j \frac{2 \pi^2 p e^{(1-2p) \pi k j}}{\sin(\pi k)}
\end{equation}
and defining $\delta=\pi k$, we arrive at

\begin{equation}
  \frac{\partial \overline{\hat{z_{1}}}}{\partial t} = \frac{i \pi \kappa}{2 a^2} \left( \left[ 2 p (1 - p) - \frac{1}{\sin^2(\delta)} \right] \hat{z_{1}} + \left[ \frac{e^{-2 p \delta j}}{\sin^2(\delta)} + j \frac{2 p e^{(1-2p) \delta j}}{\sin(\delta)} \right]\hat{z_{2}}  \right)
\end{equation}
\begin{equation}
  \frac{\partial \overline{\hat{z_{2}}}}{\partial t} = \frac{i \pi \kappa}{2 a^2} \left( \left[ - \frac{e^{2 p \delta j}}{\sin^2(\delta)} + j \frac{2 p e^{-(1-2p) \delta j}}{\sin(\delta)} \right] \hat{z_{1}}  + \left[ - 2 p (1 - p) + \frac{1}{\sin^2(\delta)} \right] \hat{z_{2}}  \right) .
\end{equation}
Writing in vector form
\begin{equation}
  \mathbf{\hat{z}} = 
  \begin{bmatrix}
    \hat{z_{1}} \\
    \hat{z_{2}}
  \end{bmatrix}
\end{equation}
the system is
\begin{equation}
  \frac{\partial \overline{\mathbf{\hat{z}}}}{\partial t} = \mathbf{M} \, \mathbf{\hat{z}}
  \label{eq:Az}
\end{equation}
where
\begin{equation}
  \mathbf{M} = \frac{i \pi \kappa}{2 a^2}
  \begin{bmatrix}
    2 p (1 - p) - \frac{1}{\sin^2(\delta)} & \frac{e^{-2 p \delta j}}{\sin^2(\delta)} + j \frac{2 p e^{(1-2p) \delta j}}{\sin(\delta)} \\
    - \frac{e^{2 p \delta j}}{\sin^2(\delta)} + j \frac{2 p e^{-(1-2p) \delta j}}{\sin(\delta)} & - 2 p (1 - p) + \frac{1}{\sin^2(\delta)} 
  \end{bmatrix} .
  \label{eq:M}
\end{equation}

The full linear system in the bicomplex formulation is then
\begin{equation}
  \frac{\partial }{\partial t} \begin{bmatrix}
    \mathbf{\hat{z}} \\
    \overline{\mathbf{\hat{z}}}
  \end{bmatrix}
   =
  \begin{bmatrix}
    0            & \overline{\mathbf{M}} \\
    \mathbf{M} & 0
  \end{bmatrix}
  \begin{bmatrix}
    \mathbf{\hat{z}} \\
    \overline{\mathbf{\hat{z}}}
  \end{bmatrix} .
  \label{eq:generalizedsystem}
\end{equation}

This relationship is the generalized linearized dynamical system, valid for every value of $d$. The classical results from literature for the symmetric and staggered von Kármán vortex streets are just particular cases of this generalized formulation. In sections~\ref{sec:stability}\ref{sec:stabilitykarman}(\ref{sec:stabilitykarmansymmetric}) and (\ref{sec:stabilitykarmanstaggered}), using the tools and approaches detailed in section~\ref{sec:framework}, it is shown that equations~\ref{eq:Az} and~\ref{eq:generalizedsystem} give the same results from the literature when using $\mathbf{M}$ defined in equation~\ref{eq:M}.

The two known steady cases are the symmetric and staggered von Kármán vortex streets. But this relationship would also be valid for other values of $d$. For example, if fixed vortex and source sheets are included between the two rows of vortices, a static configuration is possible for every value of $d$, however, this would be a mathematically constructed problem that does not have a corresponding real-world flow known by these authors. Nevertheless, part of the asymmetric reverse vortex street, formed on certain conditions by oscillating airfoils~\cite{jones1998experimental,godoy2009model,dynnikova2021stability}, resembles an inclined vortex street with $d \ne (1/2 + ih)a$. This problem is not treated here, though.


These equations are also applicable to two rows of sources. In the complex representation of potential flows, a source can be considered a vortex with imaginary circulation and a vortex can be considered a source with imaginary strength. Hence, the dynamical system that represents two rows of sources can be directly obtained by considering an imaginary value for $\kappa$.

The bicomplex formulation is paramount to arriving at equation~\ref{eq:M}, as can be inferred from the presence of complex units $j$ and $i$ (explicitly and in $\delta$). This generalized formulation of the stability of the von Kármán vortex street, to the best of the authors' knowledge, is first presented in this work.

\subsubsection{Symmetric von Kármán vortex street} \label{sec:stabilitykarmansymmetric}

First, let's compare the system to the results of a symmetric von Kármán vortex street. In this case, $\delta = \pi k=i \pi h$, where $h \in \mathbb{R}$:
\begin{equation}
  \mathbf{M} = \frac{i \pi \kappa}{2 a^2}
  \begin{bmatrix}
    2 p (1 - p) + \frac{1}{\sinh^2(\pi h)} & - \frac{e^{-2 p \pi h i j}}{\sinh^2(\pi h)} - i j \frac{2 p e^{(1-2p) \pi h i j}}{\sinh(\pi h)} \\
    \frac{e^{2 p \pi h i j}}{\sinh^2(\pi h)} - i j \frac{2 p e^{-(1-2p) \pi h i j}}{\sinh(\pi h)} & - 2 p (1 - p) - \frac{1}{\sinh^2(\pi h)} 
  \end{bmatrix}
\end{equation}
\begin{equation}
  Re^i(\mathbf{M}) = \frac{\pi \kappa}{2 a^2}
  \begin{bmatrix}
    0 & - j \frac{\sinh{(2 p \pi h)}}{\sinh^2(\pi h)} + j \frac{2 p \cosh{((1-2p) \pi h)}}{\sinh(\pi h)} \\
    -j \frac{\sinh{(2 p \pi h)}}{\sinh^2(\pi h)} + j \frac{2 p \cosh{((1-2p) \pi h)}}{\sinh(\pi h)} & 0 
  \end{bmatrix}
\end{equation}
\begin{equation}
  Im^i(\mathbf{M}) = \frac{\pi \kappa}{2 a^2}
  \begin{bmatrix}
    2 p (1 - p) + \frac{1}{\sinh^2(\pi h)} & - \frac{\cosh{(2 p \pi h)}}{\sinh^2(\pi h)} - \frac{2 p \sinh{((1-2p) \pi h)}}{\sinh(\pi h)} \\
    \frac{\cosh{(2 p \pi h)}}{\sinh^2(\pi h)} + \frac{2 p \sinh{((1-2p) \pi h)}}{\sinh(\pi h)} & - 2 p (1 - p) - \frac{1}{\sinh^2(\pi h)} 
  \end{bmatrix}
\end{equation}

Following~\cite{lamb1932hydrodynamics}, we denote
\begin{equation}
  A = 2 p (1 - p) + \frac{1}{\sinh^2(\pi h)}
\end{equation}
\begin{equation}
  B =  j \left[ - \frac{\sinh{(2 p \pi h)}}{\sinh^2(\pi h)} + \frac{2 p \cosh{((1-2p) \pi h)}}{\sinh(\pi h)} \right]
\end{equation}
\begin{equation}
  C =  -\frac{\cosh{(2 p \pi h)}}{\sinh^2(\pi h)} - \frac{2 p \sinh{((1-2p) \pi h)}}{\sinh(\pi h)} .
\end{equation}

The symmetry of the terms $B$ and $C$, formed by functions $\cosh$ and $\sinh$, was already evident in~\cite{lamb1932hydrodynamics}. However, using bicomplex numbers we can interpret these terms as the $Re^i$ and $Im^i$ parts of exponential functions, noting the following properties of hyperbolic numbers:
\begin{equation}
  \begin{split}
  e^{\theta ij} = \cosh{\theta} + ij \sinh{\theta} \\
  e^{-\theta ij} = \cosh{\theta} - ij \sinh{\theta}
  \end{split}
\end{equation}

The system in the $\mathbb{C}(j)$ formulation then becomes
\begin{equation}
  \frac{\partial }{\partial t} \begin{bmatrix}
    \hat{x}_1 \\
    \hat{x}_2 \\
    \hat{y}_1 \\
    \hat{y}_2 
  \end{bmatrix}
  = \frac{\pi \kappa}{2 a^2}
  \begin{bmatrix}
     0 &  B & -A & -C \\
     B &  0 &  C &  A \\
    -A & -C &  0 & -B \\
     C &  A & -B &  0 
  \end{bmatrix}
  \begin{bmatrix}
    \hat{x}_1 \\
    \hat{x}_2 \\
    \hat{y}_1 \\
    \hat{y}_2 
  \end{bmatrix}
  \label{eq:Cjsystemsymmetric}
\end{equation}
which is the same result presented by~\cite{lamb1932hydrodynamics} when considering the different reference systems. The system written in the bicomplex formulation is
\begin{equation}
  \frac{\partial }{\partial t} \begin{bmatrix}
    \hat{z}_1 \\
    \hat{z}_2 \\
    \overline{\hat{z}_1} \\
    \overline{\hat{z}_2} 
  \end{bmatrix}
  = \frac{\pi \kappa}{2 a^2}
  \begin{bmatrix}
      0 &    0 &   -iA &  B-iC \\
      0 &    0 &  B+iC &    iA \\
     iA & B+iC &     0 &     0 \\
   B-iC &  -iA &     0 &     0 
  \end{bmatrix}
  \begin{bmatrix}
    \hat{z}_1 \\
    \hat{z}_2 \\
    \overline{\hat{z}_1} \\
    \overline{\hat{z}_2} 
  \end{bmatrix} .
\end{equation}

As expected, it can be confirmed that both formulations give the same set of eigenvalues
\begin{equation}
  \lambda = \frac{\pi \kappa}{2 a^2} (\pm B  \pm \sqrt{A^2-C^2})
\end{equation}
where the term $(A^2-C^2)$ is always positive, as can be seen in figure~\ref{fig:A2-C2}(a), indicating that the flow is always unstable.

\begin{figure}[t]
  \centering
  \sbox0{\includegraphics[height=0.27\textwidth]{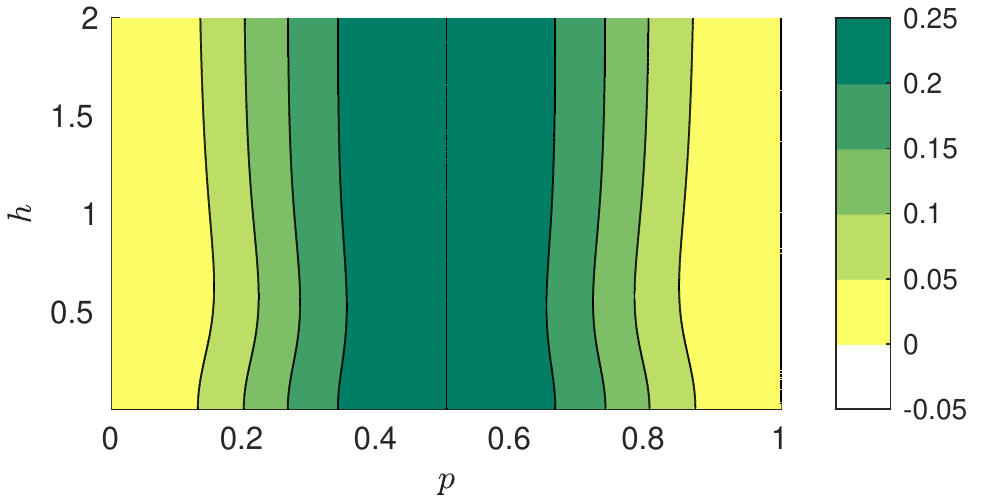}}
  \begin{tikzpicture}
    \node[anchor=south west,inner sep=0] (image) at (0,0) {\includegraphics[clip,trim={.0\wd0} {.0\ht0} {.25\wd0} {.0\ht0},height=0.27\textwidth]{Figures/symm2.pdf}};
    \node at (0.2,0.2) {(a)}; 
  \end{tikzpicture}
  \begin{tikzpicture}
    \node[anchor=south west,inner sep=0] (image) at (0,0) {\includegraphics[clip,trim={.0\wd0} {.0\ht0} {.0\wd0} {.0\ht0},height=0.27\textwidth]{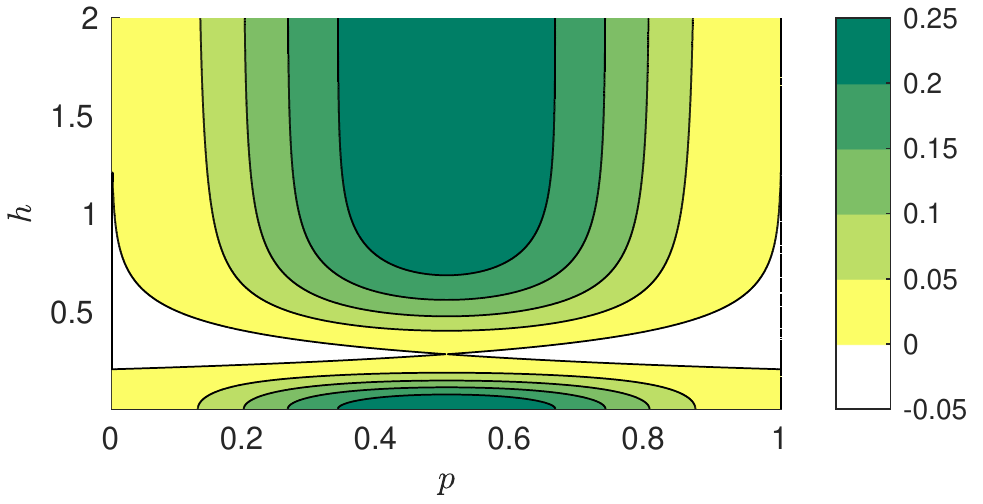}};
    \node at (0.2,0.2) {(b)}; 
  \end{tikzpicture}
  \caption{Contours of $(A^2-C^2)$. $(A^2-C^2)$ is positive for the symmetric vortex street (a) for all values of $p$ and $h$. However, $(A^2-C^2)$ can be negative for the staggered vortex street (b).}
  \label{fig:A2-C2}
\end{figure}

\subsubsection{Staggered von Kármán vortex street} \label{sec:stabilitykarmanstaggered}

Now, we compare to the results of a staggered von Kármán vortex street. In this case, $\delta = \pi k=\pi (i h+1/2)$:
\begin{equation}
  \mathbf{M} = \frac{i \pi \kappa}{2 a^2}
  \begin{bmatrix}
    2 p (1 - p) - \frac{1}{\cosh^2(\pi h)} & \frac{e^{-2 p \pi h ij} e^{- p \pi j}}{\cosh^2(\pi h)} - \frac{2 p e^{(1-2p) \pi h ij} e^{-p \pi j}}{\cosh(\pi h)} \\
    - \frac{e^{2 p \pi h ij} e^{p \pi j}}{\cosh^2(\pi h)} + \frac{2 p e^{-(1-2p) \pi h ij} e^{p \pi j}}{\cosh(\pi h)} & - 2 p (1 - p) + \frac{1}{\cosh^2(\pi h)}
  \end{bmatrix}
\end{equation}
that leads to $Re^i$ and $Im^i$ parts
\begin{equation}
  Re^i(\mathbf{M}) = \frac{\pi \kappa}{2 a^2}
  \begin{bmatrix}
    0 & B e^{-p \pi j}\\
    B e^{p \pi j} & 0
  \end{bmatrix}
\end{equation}
\begin{equation}
  Im^i(\mathbf{M}) = \frac{\pi \kappa}{2 a^2}
  \begin{bmatrix}
    A & C e^{- p \pi j} \\
    - C e^{p \pi j} & - A
  \end{bmatrix}
\end{equation}
where
\begin{equation}
  A = 2 p (1 - p) - \frac{1}{\cosh^2(\pi h)}
\end{equation}
\begin{equation}
  B =  j \left[ \frac{\sinh{(2 p \pi h)}}{\cosh^2(\pi h)} + \frac{2 p \sinh{((1-2p) \pi h)}}{\cosh(\pi h)} \right]
  \label{eq:Bstaggered}
\end{equation}
\begin{equation}
  C =  \frac{\cosh{(2 p \pi h)}}{\cosh^2(\pi h)} - \frac{2 p \cosh{((1-2p) \pi h)}}{\cosh(\pi h)} .
\end{equation}

The system in $\mathbb{C}(j)$ formulation is
\begin{equation}
  \frac{\partial }{\partial t} \begin{bmatrix}
    \hat{x}_1 \\
    \hat{x}_2 \\
    \hat{y}_1 \\
    \hat{y}_2 
  \end{bmatrix}
  = \frac{\pi \kappa}{2 a^2}
  \begin{bmatrix}
     0             &  B e^{-p \pi j} & -A             & -C e^{-p \pi j} \\
     B e^{p \pi j} &  0              &  C e^{p \pi j} &  A              \\
    -A             & -C e^{-p \pi j} &  0             & -B e^{-p \pi j} \\
     C e^{p \pi j} &  A              & -B e^{p \pi j} &  0              
  \end{bmatrix}
  \begin{bmatrix}
    \hat{x}_1 \\
    \hat{x}_2 \\
    \hat{y}_1 \\
    \hat{y}_2 
  \end{bmatrix}
  \label{eq:systemstaggered}
\end{equation}
which, again, is the same result presented by~\cite{lamb1932hydrodynamics} when considering the different definitions and reference systems. In particular, the terms $e^{-p \pi j}$ and $e^{p \pi j}$ are due to different origins of the reference system in the $x$ direction: while we consider the origin in the second vortex row to be on the first of its vortices, Lamb\cite{lamb1932hydrodynamics} considered the origin to be the same of the first vortex row. Defining $\hat{z}_2^l = \hat{z}_2 e^{-p \pi j}$, the origin of the system is moved and the perturbation in the second row is $z_{2n}' = \hat{z_{2}} e^{2 \pi p n j} = \hat{z}_2^l e^{2 \pi p (n+1/2) j}$. Then it is easy to see that equation~\ref{eq:systemstaggered} is equivalent to the system:
\begin{equation}
  \frac{\partial }{\partial t} \begin{bmatrix}
    \hat{x}_1 \\
    \hat{x}_2^l \\
    \hat{y}_1 \\
    \hat{y}_2^l 
  \end{bmatrix}
  = \frac{\pi \kappa}{2 a^2}
  \begin{bmatrix}
     0 &  B & -A & -C \\
     B &  0 &  C &  A \\
    -A & -C &  0 & -B \\
     C &  A & -B &  0 
  \end{bmatrix}
  \begin{bmatrix}
    \hat{x}_1 \\
    \hat{x}_2^l \\
    \hat{y}_1 \\
    \hat{y}_2^l 
  \end{bmatrix}
\end{equation}
which is the same system as~\ref{eq:Cjsystemsymmetric}, as expected by the results of~\cite{lamb1932hydrodynamics} (also agrees with~\cite{saffman1992vortex}, noting that~\cite{saffman1992vortex} has a typo in the equation equivalent to~\ref{eq:Bstaggered}\cite{mowlavi2016spatio}).

The system in bicomplex formulation and eigenvalues are analogous to the ones presented in section~\ref{sec:stabilitykarman}(\ref{sec:stabilitykarmansymmetric}). The eigenvalues
\begin{equation}
  \lambda = \frac{\pi \kappa}{2 a^2} (\pm B  \pm \sqrt{A^2-C^2})
  \label{eq:eigenvaluesvonkarman}
\end{equation}
can have the term $\sqrt{A^2-C^2}$ not real (figure~\ref{fig:A2-C2}(b)). In this case, it should be defined so that it belongs to $\mathbb{C}(j)$, as discussed in section~\ref{sec:framework}. In other words, if $C^2>A^2$, then $\pm \sqrt{A^2-C^2}$ should be represented by $\pm j\sqrt{C^2 - A^2}$ (not $\pm i\sqrt{C^2 - A^2}$). Hence, equation~\ref{eq:eigenvaluesvonkarman} is equivalent to
\begin{equation}
  \lambda = \frac{\pi \kappa}{2 a^2} (\pm B  \pm j\sqrt{C^2-A^2}) .
\end{equation}

It can be seen in figure~\ref{fig:A2-C2}(b) that the stability to infinitesimal perturbations of every wavenumber $p$ occurs only for one possible value of $h$. This is the famous value of $h$ calculated by~\cite{karman1912mechanismus} that can be found by solving the inequality $C^2-A^2 \geq 0$ for $p=1/2$:
\begin{equation}
  h = \frac{1}{\pi} \cosh^{-1}{\sqrt{2}} \approx 0.2805 .
\end{equation}

These results show that the development performed in bicomplex variables gives results equivalent to the derivation in terms of a single complex unit. Moreover, the linear system in the bicomplex space can be interpreted as a generalized solution of the von Kármán vortex street, for which the symmetric and staggered vortex streets are particular cases.

\subsection{Stability of a single row of vortices} \label{sec:stabilityrow}
The case of a single row of vortices is a simplification of the von Kármán vortex street. In this case
\begin{equation}
  \mathbf{M} = \frac{i \pi \kappa}{2 a^2}
  \begin{bmatrix}
    2 p (1 - p)
  \end{bmatrix} ,
\end{equation}
\begin{equation}
  A = 2 p (1 - p) ,
\end{equation}
the system in the $\mathbb{C}(j)$ formulation becomes
\begin{equation}
  \frac{\partial }{\partial t} \begin{bmatrix}
    \hat{x}_1 \\
    \hat{y}_1 
  \end{bmatrix}
  = \frac{\pi \kappa}{2 a^2}
  \begin{bmatrix}
     0 & -A \\
    -A &  0 
  \end{bmatrix}
  \begin{bmatrix}
    \hat{x}_1 \\
    \hat{y}_1 
  \end{bmatrix}
\end{equation}
and the system in bicomplex formulation is
\begin{equation}
  \frac{\partial }{\partial t} \begin{bmatrix}
    \hat{z}_1 \\
    \overline{\hat{z}_1} 
  \end{bmatrix}
  = \frac{\pi \kappa}{2 a^2}
  \begin{bmatrix}
     0 & -i A \\
    i A &  0 
  \end{bmatrix}
  \begin{bmatrix}
    \hat{z}_1 \\
    \overline{\hat{z}_1} 
  \end{bmatrix} .
\end{equation}

The eigenvalues are then
\begin{equation}
  \lambda = \pm \frac{\pi \kappa}{2 a^2} A = \pm \frac{\pi \kappa}{a^2} p (1 - p)
\end{equation}
just as calculated by~\cite{saffman1992vortex} and~\cite{lamb1932hydrodynamics} (when considering the different definitions and reference systems). The eigenvalues are all real, indicating exponential growth or decay without oscillatory behaviour. This relationship has been shown to be relevant for the stability of helical vortices, theoretically, numerically and experimentally~\cite{sarmast2014mutual,quaranta2015long,quaranta2019local}. In particular, the maximum growth rate $\sigma_{max} = \max{\{Re^j(\lambda)\}}$ of disturbances is found for $p=1/2$
\begin{equation}
  \frac{2 a^2}{\kappa} \sigma_{max} = \frac{\pi}{2}
\end{equation}
which corresponds to the out-of-phase motion of neighboring vortices.

The eigenvectors of $\mathbf{R}$ (eigenvectors in terms of $\hat{x}_1$ and $\hat{y}_1$) are
\begin{equation}
  \begin{bmatrix}
    1 \\
    1 
  \end{bmatrix}  \text{and}
  \begin{bmatrix}
    1 \\
    -1 
  \end{bmatrix}
\end{equation}
and the eigenvectors of $\mathbf{Q}$ (eigenvectors in terms of $\hat{z}_1$ and $\overline{\hat{z}_1}$) are
\begin{equation}
  \begin{bmatrix}
    1 \\
    -i 
  \end{bmatrix}  \text{and}
  \begin{bmatrix}
    1 \\
    i 
  \end{bmatrix} .
\end{equation}

Saffman\cite{saffman1992vortex} arrived at a second-order system in terms of $\hat{z}_1$. As shown in theorem~\ref{the:eigMconjM}, the square root of the eigenvalues of the second-order system in terms of $\hat{z}_1$ gives the eigenvalue of the first-order system.

For this case, the second-order system in terms of $\hat{z}_1$ and $\overline{\hat{z}_1}$ is
\begin{equation}
  \frac{\partial^2 }{\partial t^2} \begin{bmatrix}
    \hat{z}_1 \\
    \overline{\hat{z}_1} 
  \end{bmatrix}
  = \left( \frac{\pi \kappa}{2 a^2} \right)^2
  \begin{bmatrix}
     A^2 & 0    \\
     0   &  A^2 
  \end{bmatrix}
  \begin{bmatrix}
    \hat{z}_1 \\
    \overline{\hat{z}_1} 
  \end{bmatrix}
\end{equation}
and the system in terms of $\hat{z}_1$ is
\begin{equation}
  \frac{\partial^2 \hat{z}_1}{\partial t^2}
  = \left( \frac{\pi \kappa}{2 a^2} \right)^2 A^2 \hat{z}_1 .
\end{equation}
As expected, both formulations give the eigenvalues
\begin{equation}
  \mu = \left( \frac{\pi \kappa}{2 a^2} A \right)^2 = \lambda^2 .
\end{equation}
We can also confirm that the eigenvectors of the second-order system do not give information about the eigenvectors of the first-order system: every vector is an eigenvector of $\mathbf{Q}^2$ because $\mathbf{Q}^2$ is a multiple of the identity matrix, while the eigenvectors of $\mathbf{Q}$ are well defined.

Saffman\cite{saffman1992vortex} however, used $i$ to represent both the complex velocity and the complex exponential \emph{ansatz}. Basically, the only difference from the current method was the application of an \emph{ansatz} in the form $e^{2 \pi p n i}$ in step~\ref{item:complexansatz} of section~\ref{sec:framework}. Analysing the results, we can understand why this formal error is inconsequential for this particular case:
\begin{itemize}
  \item Matrix $\mathbf{M}$ and, consequently, $\mathbf{R}$, $\mathbf{Q}$ and $\mathbf{Q}^2$ do not have components in $j$;
  \item The eigenvalues are completely real, hence there is no ambiguity regarding the interpretation of a complex eigenvalue;
  \item The eigenvectors in terms of $\hat{x}_1$ and $\hat{y}_1$ are completely real. The eigenvectors in terms of $\hat{z}_1$ and $\overline{\hat{z}_1}$ are complex, however, the imaginary part was correctly interpreted by Saffman\cite{saffman1992vortex}, as the $y$-component of the eigenvector, without considering that it may be interpreted as a phase difference. Nonetheless, if there were a phase difference between $\hat{x}_1$ and $\hat{y}_1$, in other words, if one of the terms had an imaginary part $j$, the interpretation could be compromised. For example, let's assume that one eigenvector of $\mathbf{R}$ for a certain problem is 
  \begin{equation}
    \begin{bmatrix}
      1 \\
      j \\
    \end{bmatrix}
  \end{equation}
  in our notation. The corresponding eigenvector of $\mathbf{Q}$ would be (applying corollary~\ref{cor:eigenvectorsRQ} and then normalizing):
  \begin{equation}
    \begin{bmatrix}
      \frac{1+ij}{2} \\
      \frac{1-ij}{2} \\
    \end{bmatrix}
  \end{equation}
  that would appear as
  \begin{equation}
    \begin{bmatrix}
      1 \\
      0 \\
    \end{bmatrix}
  \end{equation}
  if no distinction between $i$ and $j$ is made (as in the method of~\cite{saffman1992vortex}), which would lead to errors. Hence, the fact that the eigenvectors of $\mathbf{R}$ are completely real for a row of vortices, means that the interpretation of the imaginary part of the eigenvectors of $\mathbf{Q}$ as the $y$-component is correct.
\end{itemize}




\section{Extension of the method}
This framework can be applied to other flow configurations, since it is believed to be applicable to all two-dimensional incompressible potential flows. By applying conformal maps, more complex geometries can be studied. For example, airfoils can be represented in a relatively simpler manner by Joukowski or Kármán-Trefftz transformations~\cite{milne1968theoretical}. Another example is a direct extension of the configuration studied here: the von Kármán vortex street of hollow vortices, whose equilibrium solutions have been found~\cite{crowdy2011analytical}, however, their linear stability is yet to be calculated (as far as the authors are aware).

A natural extension of the framework presented here is the study of spatio-temporal stability (similarly to~\cite{mowlavi2016spatio}), by imposing disturbances with \emph{ansatz} in the form $e^{j(\varphi n-\omega t)}$, where $\varphi$ can be complex in $\mathbb{C}(j)$. The use of multiple imaginary units in the study of the response of flows to harmonic excitation, in the form $e^{j \omega t}$, has already been performed in~\cite{wu1961swimming,moore2014analytical,moore2017fast,hauge2021new,baddoo2021generalization}. The formalism provided by the bicomplex algebra might facilitate the adoption of this approach.

This framework was applied to a case of a potential flow. However, the same ideas can be used in any problem in which the $y$-component is represented by the imaginary part and the $x$-component is represented by the real part of a complex position, velocity or field. The main concept applied here is to represent the space or field with one complex unit and the \emph{ansatz} or kernel with a different complex unit. Multiple imaginary units have already been applied in the context of transform methods for the complex Helmholtz equation~\cite{hauge2021new}.

\section{Conclusions}
In this work, we present a framework that reconciles the complex representation of potential flows and the complex representation of perturbations in the study of the stability of two-dimensional incompressible potential flows. The same problems can be solved with real vectors and matrices, as has been done in the past. However, this is also the case for other applications of complex numbers in classical mechanics: there is a real matrix form of the complex numbers that allows problems to be solved using the algebra of real matrices. Nevertheless, the complex representation is usually preferred due to its convenient and compact formulation. The same can be said of bicomplex numbers. By trying to solve the stability problem, we noted that the use of bicomplex numbers is a natural way to unify the two representations, each by a different imaginary unity. Known challenges of the bicomplex algebra, such as the existence of multiple zero divisors and the loss of the Fundamental Theory of Algebra, were showed not to be issues for the current application. The application of this method to the von Kármán vortex street allowed the discovery of a generalized linear system, showing that the convenience and compactness of the bicomplex approach are worthwhile.

This paper was born out of a mistake. The first author, naively, tried to apply the method of chapter 7.5 of~\cite{saffman1992vortex} to the von Kármán vortex street. Many other mistakes or incomplete results can also be found in the literature on the subject (see section~\ref{sec:introduction}). If this paper deters other people from making the same mistakes, it is already worth the reading. However, we argue that the power of the framework presented here goes way beyond that. By applying this method to a century-old problem that is one of the most studied flow phenomena, we found a generalized formula that was previously unknown. Maybe some analytical solutions for more sophisticated flows are just waiting to be discovered by using the mathematical tools allowed by the algebra of bicomplex numbers. At this moment, it is reasonable to update the quote of Munk\cite{munk1925elements}: for the stability of two-dimensional potential flows, it might pay to get acquainted with this method even if one has never occupied himself with bicomplex numbers before.

\enlargethispage{20pt}

\dataccess{This article has no additional data.}

\aucontribute{V.K.: conceptualization, methodology, formal analysis, writing—original draft. A.H. and D.H.: project administration, supervision, writing—review and editing.}

\competing{The authors declare no competing interests.}

\funding{VK thanks KTH Engineering Mechanics for partially funding this work.}



\vskip2pc


\bibliographystyle{RS} 
\bibliography{references} 

\end{document}